\documentclass[a4paper,UKenglish,cleveref, autoref, thm-restate]{lipics-v2021}

\hideLIPIcs  

\title{Determining the Outerthickness of Graphs Is NP-Hard}

\author{Pin-Hsian Lee}{National Taiwan University, Taiwan}{40940121S@ntnu.edu.tw}{https://orcid.org/0009-0009-2677-4614}{}
\author{Te-Cheng Liu}{Academia Sinica, Taiwan}{david1213@sinica.edu.tw}{https://orcid.org/0009-0002-3031-3238}{}
\author{Meng-Tsung Tsai}{Academia Sinica, Taiwan}{mttsai@iis.sinica.edu.tw}{https://orcid.org/0000-0002-2243-8666}{This research was supported in part by the National Science and Technology Council under
contract NSTC 114-2221-E-001-023.}

\authorrunning{P.-H. Lee, T.-C. Liu, and M.-T. Tsai} 

\Copyright{Pin-Hsian Lee, Te-Cheng Liu, and Meng-Tsung Tsai} 



\category{} 

\relatedversion{} 

\nolinenumbers 

\usepackage{tikz}
\usetikzlibrary{shapes, arrows, positioning, calc}
\usetikzlibrary{calc, positioning, shapes.geometric, decorations.pathmorphing, fit}
\usepackage{mathtools}
\usepackage{todonotes}
\usepackage{booktabs}

\hypersetup{colorlinks=true,
 linkcolor=red,
citecolor=green,
urlcolor=magenta
}
\definecolor{purple}{HTML}{1A1AB3}
\definecolor{blue}{HTML}{2F6FED}
\definecolor{red}{HTML}{D32F2F}
\definecolor{green}{HTML}{2E7D32}

\DeclareEmphSequence{\color{purple}\itshape}

\AtBeginEnvironment{algorithm}{%
  \DeclareEmphSequence{\itshape}
}
\AtEndEnvironment{algorithm}{%
  \DeclareEmphSequence{\color{purple}\itshape}
}

\AtBeginEnvironment{thebibliography}{%
  \DeclareEmphSequence{\itshape}
}
\AtEndEnvironment{thebibliography}{%
  \DeclareEmphSequence{\color{purple}\itshape}
}

\newcommand{\yes}{Y}
\newcommand{\no}{{\color{red}N}}
\newcommand{\othersNP}[1]{NP-hard~#1}
\newcommand{\ourNP}{{\color{blue} NP-hard}}

\usepackage{amsthm}
\theoremstyle{definition}
\newtheorem{problem}[theorem]{Problem}

\EventEditors{}
\EventNoEds{0}
\EventLongTitle{}
\EventShortTitle{ICALP}
\EventAcronym{}
\EventYear{}
\EventDate{}
\EventLocation{}
\EventLogo{}
\SeriesVolume{}
\ArticleNo{1}

\begin{document}

\maketitle

\begin{abstract}
We give a short, self-contained, and easily verifiable proof that determining the outerthickness of a
general graph is NP-hard. This resolves a long-standing open problem on the computational complexity
of outerthickness.

Moreover, our hardness result applies to a more general covering problem $P_{\mathcal{F}}$, defined as
follows. Fix a proper graph class $\mathcal{F}$ whose membership is decidable. Given an
undirected simple graph $G = (V, E)$ and an integer $k$, the task is to cover the edge set $E(G)$ by at most $k$ subsets
$E_1,\ldots,E_k$ such that each subgraph $(V(G),E_i)$ for $i \in [k]$ belongs to $\mathcal{F}$. Note that if
$\mathcal{F}$ is monotone (in particular, when $\mathcal{F}$ is the class of all outerplanar graphs),
any such cover can be converted into an edge partition by deleting overlaps; hence, in this case,
covering and partitioning are equivalent.

Our result shows that for every proper graph class $\mathcal{F}$ whose membership is decidable and that satisfies all of the following conditions:
(a)~$\mathcal{F}$ is closed under topological minors,
(b)~$\mathcal{F}$ is closed under $1$-sums, and
(c)~$\mathcal{F}$ contains a cycle of length $3$,
the problem $P_{\mathcal{F}}$ is NP-hard for every fixed integer $k\ge 3$. In particular:

\begin{itemize}
    \item For $\mathcal{F}$ equal to the class of all outerplanar graphs, our result settles the
    long-standing open problem on the complexity of determining outerthickness.
    \item For $\mathcal{F}$ equal to the class of all planar graphs, our result complements Mansfield's NP-hardness result (1983) for the thickness, which applies only to the case $k=2$.
\end{itemize}

It is also worth noting that each of the three conditions above is necessary. If $\mathcal{F}$ is the
class of all eulerian graphs, then condition~(a) fails. If $\mathcal{F}$ is the class of all pseudoforests,
then condition~(b) fails. If $\mathcal{F}$ is the class of all forests, then condition~(c) fails. For each
of these three classes $\mathcal{F}$, the problem $P_{\mathcal{F}}$ is solvable in polynomial time for every fixed integer 
$k\ge 3$, showing that none of the three conditions can be dropped.
\end{abstract}

\thispagestyle{empty}
\clearpage
\setcounter{page}{1}
\section{Introduction}\label{sec:intro}

Guy (1974)~\cite{Guy74} defined the \emph{outerthickness} of a graph $G$ as the minimum number of parts in an edge partition of $G$ such that each part induces an outerplanar subgraph; 
that is, a planar graph admitting an embedding in the plane in which every node lies on the outer face. Outerthickness is an analogue of thickness, except that each part is required to be outerplanar rather than merely planar. Although Mansfield (1983) proved that determining the thickness of a graph is NP-hard~\cite{Mansfield83}, the complexity of determining the outerthickness has remained unsettled~\cite{Poranen04,PoranenM05,Poranen08,BatraGPC10,MakinenP12,SuZ21,HlinenyM23,BalkoHMOVW24}. 

The algorithmic study of outerthickness is motivated in part by a conjecture of Chartrand et al. (1971), who proposed, among other things, that every planar graph admits an edge partition into two outerplanar graphs; equivalently, that every planar graph has outerthickness at most 2. Elmallah and Charles (1988) cited this conjecture and addressed it for 4-node-connected planar graphs by presenting an algorithm that produces such a partition into two outerplanar subgraphs. Although they did not explicitly analyze the running time, using the state of the art of the required subroutines at the time~\cite{HopcroftT74,Gouyou-Beauchamps82}, their algorithm can be implemented in $O(n^3)$ time, where $n$ denotes the number of nodes in the input graph (and we use this convention throughout the paper).

Subsequently, Heath (1991) claimed a proof of Chartrand et al.’s conjecture and gave an $O(n^2 \log n)$-time algorithm to partition the edges of any planar graph into at most two outerplanar graphs~\cite{Heath91}. Later, Gon{\c{c}}alves (2005) established the same statement with an optimal $O(n)$-time construction~\cite{Goncalves05}. It is worth noting that Gon{\c{c}}alves remarked that some earlier claimed proofs were later found to be incorrect, and subsequent work attribute the resolution of Chartrand et al.’s conjecture to Gon{\c{c}}alves~\cite{DujmovicWood07,XuZha18,DiGiacomo18,BekosEtAl19}. 

Unlike the substantial progress on planar graphs, there has been little progress on determining outerthickness for general graphs. M\"{a}kinen and Poranen (2004) conjectured in~\cite{Poranen04,PoranenM05} that deciding the outerthickness of a general graph is NP-hard. Since then, the open status of this problem has been repeatedly highlighted in subsequent work~\cite{Poranen08,BatraGPC10,MakinenP12,SuZ21}, and was reiterated as recently as 2023 by Hlinen\'{y} and Masar\'{\i}k~\cite{HlinenyM23} and 2024 by Balko et al.~\cite{BalkoHMOVW24}. We further note that Balko et al.\ prove NP-hardness for computing the uncrossed number under the assumption that computing outerthickness is NP-hard~\cite{BalkoHMOVW24}.

In this paper, we settle the computational complexity of determining the outerthickness of general graphs, a question that has remained open since outerthickness was introduced in 1974. Formally, we consider the following decision problem.

\begin{problem}[Outerthickness]
\label{prob:outerthickness}
~
\begin{itemize}
    \item \textbf{Input:} An $n$-node undirected simple graph $G=(V,E)$ and an integer $k\ge 1$.
    \item \textbf{Question:} Does $G$ admit an edge partition $E=E_1\cup E_2\cup \cdots \cup E_k$
    such that $(V,E_i)$ is outerplanar for every $i\in[k]$?
\end{itemize}
\end{problem}

We use \textsc{OuterThickness}$(G, k)$ to denote \cref{prob:outerthickness} with input graph $G$ and parameter $k$ for short. Our main result is the following.
\begin{theorem}
\label{thm:kernel}
For every fixed integer $k \ge 3$, \textsc{Outerthickness}$(G, k)$ for general graphs $G$ is NP-complete.
\end{theorem}

Moreover, our hardness result applies to a more general covering problem $P_{\mathcal{F}}$, defined as
follows. Fix a proper graph class $\mathcal{F}$ whose membership is decidable. Given an
undirected simple graph $G = (V, E)$ and an integer $k$, the task is to cover the edge set $E(G)$ by at most $k$ subsets
$E_1,\ldots,E_k$ such that each subgraph $(V(G),E_i)$ belongs to $\mathcal{F}$. Note that if
$\mathcal{F}$ is \emph{monotone}\footnote{If $G \in \mathcal{F}$ and $H$ is a subgraph of $G$, then $H \in \mathcal{F}$.
} (in particular, when $\mathcal{F}$ is the class of all outerplanar graphs),
any such cover can be converted into an edge partition by deleting overlaps; hence, in this case,
covering and partitioning are equivalent.

We generalize the NP-hardness of \textsc{OuterThickness}$(G, k)$ to $P_{\mathcal{F}}$ for every \emph{proper}\footnote{A graph class $\mathcal{F}$ is \emph{proper} if it does not contain all graphs.} $\mathcal{F}$ whose membership is decidable and that satisfies the three conditions stated below:
\begin{theorem}\label{thm:main}
For every proper graph class $\mathcal{F}$ whose membership is decidable and that satisfies all of the following conditions:
\begin{enumerate}[(a)]
    \item $\mathcal{F}$ is \emph{closed under topological minors}\footnote{A graph class $\mathcal{F}$ is \emph{closed under topological minors} if for every $G\in\mathcal{F}$
and every graph $H$ that is a topological minor of $G$, we have $H\in\mathcal{F}$. Here $H$ is a
\emph{topological minor} of $G$ if some subdivision of $H$ appears as a subgraph of $G$ (equivalently,
$H$ can be obtained from a subgraph of $G$ by repeatedly smoothing degree-$2$ nodes).},
    \item $\mathcal{F}$ is \emph{closed under $1$-sums}\footnote{A graph class $\mathcal{F}$ is \emph{closed under $1$-sums} if for every $G_1,G_2\in\mathcal{F}$
    and every choice of nodes $v_1\in V(G_1)$ and $v_2\in V(G_2)$, the graph obtained from the disjoint union $G_1 \uplus G_2$ by identifying $v_1$ and $v_2$ into a single node also belongs to $\mathcal{F}$.}, and
    \item $\mathcal{F}$ contains a  cycle of length $3$,
\end{enumerate}
the problem $P_{\mathcal{F}}$ is NP-hard for every fixed integer $k\ge 3$.
\end{theorem}
Moreover, if membership in $\mathcal{F}$ can be decided by an algorithm $M_{\mathcal{F}}$ in polynomial time, then $P_{\mathcal{F}} \in \mathrm{NP}$.
A certificate is an edge covering of $E(G)$ by $k$ subsets $E_i$ for $i\in[k]$, and we can verify in polynomial time
that $(V(G),E_i)\in\mathcal{F}$ for every $i\in[k]$ by running $M_{\mathcal{F}}$ on each part. Hence the certificate
certifies that the $\mathcal{F}$-thickness of $G$ is at most $k$. As a result, \cref{thm:main} also yields NP-completeness if membership in 
$\mathcal{F}$ can be decided in polynomial time.

As a corollary of \cref{thm:main}, by setting $\mathcal{F}$ to be the class of all planar graphs, we obtain
\cref{cor:thickness}, which complements Mansfield's NP-hardness result (1983) for the thickness~\cite{Mansfield83},
established only for the case $k=2$.

\begin{corollary}
\label{cor:thickness}
For every fixed integer $k \ge 3$, \textsc{Thickness}$(G, k)$ for general graphs is NP-complete.
\end{corollary}

It may be worth noting that Mansfield's NP-hardness proof for the thickness proceeds via a reduction
from Planar-3SAT~\cite{Lichtenstein82} due to Lichtenstein (1982), which is technical.
In contrast, our hardness proof uses a more direct sources: for $k=3$, we reduce from the
chromatic index problem for $3$-regular graphs due to Holyer (1981)~\cite{Holyer81}, and for every fixed
integer $k\ge 3$, we reduce from the chromatic index problem for $k$-regular graphs due to Leven and
Galil (1983)~\cite{LevenG83}.

\begin{table}[!ht]
\centering
\label{tab:template}
\begin{tabular}{lccccccc}
$\mathcal{F}$ & Cond. (a) & Cond. (b) & Cond. (c) & $k = 1$ & $k=2$ & each $k \ge 3$ \\
\midrule
\midrule
forests & \yes & \yes & \no & P & P~\cite{Edmonds65} & P~\cite{Edmonds65}  \\
pseudo-forests & \yes & \no & \yes & P & P~\cite{Edmonds65} & P~\cite{Edmonds65}  \\
eulerian graphs & \no & \yes & \yes & P & \othersNP{\cite{PetrusevskiS21}} & P~\cite{AlonT85}  \\
cacti & \yes & \yes & \yes & P &  & \ourNP  \\
outerplanar graphs & \yes & \yes & \yes & P~\cite{Mitchell79} & & \ourNP   \\
planar graphs & \yes & \yes & \yes & P~\cite{HopcroftT74} & \othersNP {\cite{Mansfield83}} & \ourNP  \\
graphs of $tw(G) \le t$ & \yes & \yes & \yes & P~\cite{Bodlaender96} & & \ourNP \\
\bottomrule
\end{tabular}

\smallskip
\caption{Graph classes $\mathcal{F}$. Note that  eulerian graphs here including  disconnected ones, and $tw(G)$ denotes the treewidth of $G$. The results highlighted in blue are proved in \cref{thm:main}.\label{tab:classes}}
\end{table}

As illustrated in~\cref{tab:classes}, each of the three conditions mentioned in~\cref{thm:main} is necessary. If $\mathcal{F}$ is the class of all eulerian graphs (including disconnected ones), then condition~(a) fails. If $\mathcal{F}$ is the class of all pseudoforests,
then condition~(b) fails. If $\mathcal{F}$ is the class of all forests, then condition~(c) fails. For each of these three classes $\mathcal{F}$, the problem $P_{\mathcal{F}}$ is solvable in polynomial time for every fixed integer $k \ge 3$ (the case $P_{\mathrm{eulerian}}$ is due to Alon and Tarsi~\cite{AlonT85}, and the cases $P_{\mathrm{forests}}$ and $P_{\mathrm{pseudoforests}}$ follow from Edmonds~\cite{Edmonds65}), showing that none of the three conditions can be dropped.


\subsection{Paper Organization}

In \cref{sec:prel}, we introduce our notation. In \cref{sec:kernel}, we present a minimal proof
showing that outerthickness is NP-hard, thereby settling the long-standing open question on its
computational complexity. Then, in \cref{sec:main}, we prove our main theorem (\cref{thm:main}),
which in particular establishes \cref{thm:kernel} for every fixed integer $k>3$ and strengthens the
minimal proof from \cref{sec:kernel}. We conclude with remarks in \cref{sec:con}.

\section{Preliminaries}\label{sec:prel}

All graphs and subgraphs considered in this paper are finite, undirected, and simple. For a graph $G$, we write
$V(G)$ and $E(G)$ for its node set and edge set, respectively. For convenience, we sometimes use an edge subset
$E'\subseteq E(G)$ to denote the subgraph of $G$ induced by $E'$.

Fix a graph property $\mathcal{P}$. A graph $G$ is said to be \emph{edge-maximal} with respect to  $\mathcal{P}$ if $G$ satisfies $\mathcal{P}$, but for every pair of non-adjacent nodes $u, v \in V(G)$, the graph obtained by adding the edge $\{u, v\}$ to $G$ does not satisfy property $\mathcal{P}$.

For every integer $m \ge 1$, we use the notation $[m]$ to denote the set of integers $\{1, \dots, m\}$. The reduction source of our problem is defined below.

\begin{theorem}[Edge-Coloring~\cite{Holyer81,LevenG83}]
\label{thm:edgecolor}
Fix an integer $k \ge 3$.
A \emph{proper $k$-edge-coloring} of an undirected graph $G = (V,E)$ is a mapping
$c : E \to \{1,2,\dots,k\}$ such that $c(e) \ne c(f)$ for every pair of edges
$e,f \in E$ that share an endnode.
Given a $k$-regular undirected simple graph $G$, deciding whether $G$ admits a proper
$k$-edge-coloring is NP-complete.
\end{theorem}
\section{NP-Hardness of Outerthickness}
\label{sec:kernel}

Given an input instance $G$ of \textsc{Edge-Coloring} on $3$-regular graphs, defined in~\cref{thm:edgecolor}. We compute the following label function $\mathcal{L}$ and an auxiliary graph $H$ in polynomial time. Then, we obtain a supergraph $G'$ from joining $G$ and $H$ using label function $\mathcal{L}$ such that $G$ admits a 3-edge-coloring if and only if $G'$ has an edge-partition into 3 outerplanar subgraphs (\cref{lem:kernel}). This proves the NP-hardness stated in \cref{thm:kernel} for $k = 3$. Moreover, \textsc{OuterThickness}$(G,3)\in \mathrm{NP}$: a certificate is an edge partition of $E(G)$ into $3$ subsets
$E_i$ for $i\in[3]$, and we can verify in polynomial time that $(V(G),E_i)$ is outerplanar for every $i\in[3]$ by
running the linear-time outerplanarity test of Mitchell~\cite{Mitchell79} on each part. Hence the certificate certifies
that the outerthickness of $G$ is at most $3$.

\begin{itemize}
    \item Let $\varphi: E(G)\to \mathcal{L}$ be a labeling function, where $\mathcal{L} = \{1, 2, \ldots, |\mathcal{L}|\}$ is a constant-size (not necessarily minimum-size) label set, such that for every path of length at most $3$ in $G$, the edges on the path receive pairwise distinct labels under $\varphi$. The existence of $\varphi$, as well as a construction running in time polynomial in the input size of $G$, is guaranteed by \cref{pro:varphi}.

\medskip

    \item We pick a sufficiently large constant $C$ such that every $C$-node graph with outerthickness $3$ contains an (not necessarily maximum) independent set of $\alpha \ge |\mathcal{L}|$ nodes. Then we construct a $C$-node graph $H$ that is edge-maximal with outerthickness $3$; that is, adding any edge not in $E(H)$ increases the outerthickness of $H$. Given the promise of $C$, let $w_1, w_2, \ldots, w_\alpha$ be the $\alpha$ nodes in an independent set in $H$. The existence of $H$, as well as a construction running in time polynomial in the input size of $G$, is guaranteed by \cref{pro:H}.
\end{itemize}
We initialize $G'$ by taking copies of $G$ and $H$ with disjoint node sets, and then identifying $G'$ with their union. Then, for each edge $e\coloneqq\{u,v\}\in E(G)$, we add the edges $\{u,w_{\varphi(e)}\}$ and $\{v,w_{\varphi(e)}\}$ to $G'$, as illustrated in~\cref{fig:sparse_gadget}. Note that $G'$ remains a simple graph: since any two edges incident to the same node in $G$ receive distinct labels under $\varphi$, no parallel edges are created.

\begin{figure}[!h]
    \centering
    \begin{tikzpicture}[
    scale=1.0,
    vertex/.style={circle, draw, fill=black, minimum size=5pt, inner sep=0pt, line width=0.6pt},
    svertex/.style={circle, draw, fill=white, minimum size=6pt, inner sep=0pt, line width=0.6pt},
    gedge/.style={line width=0.8pt},
    t1edge/.style={line width=0.6pt, dashed},
    t2edge/.style={line width=0.6pt, dotted},
    font=\small
]

\begin{scope}[shift={(0, 3.5)}]
    \draw[line width=0.8pt, rounded corners=12pt] (-4.5, -0.8) rectangle (4.5, 0.8);

    \node[font=\small] at (-3.8, 0.4) {$H$};

    \node[svertex, label={[font=\scriptsize]above:$w_1$}] (s1) at (-2, 0) {};
    \node[svertex, label={[font=\scriptsize]above:$w_2$}] (s2) at (0, 0) {};
    \node[svertex, label=
    {[font=\scriptsize]above:$w_3$}] (s3) at (2, 0) {};
    \node at (3,0) {$\cdots$};
\end{scope}

\node[vertex] (u000) at (-3.5, 0) {};
\node[vertex] (u001) at (-2.5, 0) {};
\node[vertex] (u011) at (-1.5, 0) {};
\node[vertex] (u010) at (-0.5, 0) {};
\node[vertex] (u110) at ( 0.5, 0) {};
\node[vertex] (u111) at ( 1.5, 0) {};
\node[vertex] (u101) at ( 2.5, 0) {};
\node[vertex] (u100) at ( 3.5, 0) {};

\draw[gedge, draw=gray] (u000) -- (u001) node[midway, above, font=\scriptsize, text=gray] {$5$};
\draw[gedge, draw=gray] (u001) -- (u011) node[midway, above, font=\scriptsize, text=gray] {$4$};
\draw[gedge, draw=gray] (u011) -- (u010) node[midway, above, font=\scriptsize, text=gray] {$6$};
\draw[gedge]            (u010) -- (u110) node[midway, above, font=\scriptsize] {$2$};
\draw[gedge, draw=gray] (u110) -- (u111) node[midway, above, font=\scriptsize, text=gray] {$5$};
\draw[gedge, draw=gray] (u111) -- (u101) node[midway, above, font=\scriptsize, text=gray] {$3$};
\draw[gedge, draw=gray] (u101) -- (u100) node[midway, above, font=\scriptsize, text=gray] {$6$};

\draw[gedge, draw=gray] (u000) to[bend right=35] node[midway, below, font=\scriptsize, text=gray] {$3$} (u010); 
\draw[gedge]            (u011) to[bend right=35] node[midway, below, font=\scriptsize] {$1$} (u111); 
\draw[gedge]            (u001) to[bend right=45] node[midway, below, font=\scriptsize] {$2$} (u101); 
\draw[gedge, draw=gray] (u110) to[bend right=35] node[midway, below, font=\scriptsize, text=gray] {$4$} (u100); 
\draw[gedge]            (u000) to[bend right=55] node[pos=0.5, below, font=\scriptsize] {$1$} (u100); 

\node[font=\small] at (-4.5, 0) {$G$};

\draw[t1edge] (u000) -- (s1);
\draw[t1edge] (u100) -- (s1);
\draw[t1edge] (u011) -- (s1);
\draw[t1edge] (u111) -- (s1);

\draw[t2edge] (u001) -- (s2);
\draw[t2edge] (u101) -- (s2);
\draw[t2edge] (u010) -- (s2);
\draw[t2edge] (u110) -- (s2);

\end{tikzpicture}
    \caption{An illustration of the construction of $G'$. 
    For each edge $e\coloneqq\{u,v\}\in E(G)$, we add the edges $\{u,w_{\varphi(e)}\}$ and $\{v,w_{\varphi(e)}\}$ to $G'$.
    In this example, we have $\mathcal{L} = \{1, 2, \ldots, 6\}$, and we depict the added edges corresponding to $e$ with $\varphi(e) \in \{1,2\}$.
    \label{fig:sparse_gadget}
    }
\end{figure}
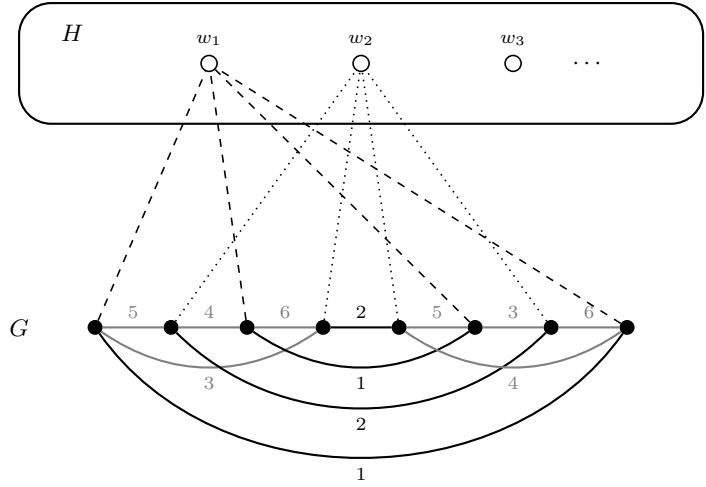

We need the following observation for our main lemma.

\begin{observation}\label{obs:extrapath}
Fix an integer $z \ge 1$. Let $Q$ be a graph that is edge-maximal with outerthickness $z$, and let $E_1, E_2, \ldots, E_z$ be an edge-partition of $E(Q)$ into $z$ outerplanar subgraphs.
Let $P$ be a path with endpoints $x$ and $x'$
such that $V(P)\cap V(Q)=\varnothing$. For any nodes $y, y'\in V(Q)$, define
\[
S_{y, y'} \coloneqq E_1 \cup P \cup \bigl\{\{x,y\},\{x',y'\}\bigr\}.
\]

\begin{enumerate}[(a)]
    \item If $y=y'$, then $S_{y,y'}$ is outerplanar.
    \item If $y\ne y'$ and $\{y,y'\}\notin E(Q)$, then $S_{y,y'}$ is not outerplanar.
\end{enumerate}
\end{observation}
\begin{proof}
Fix an outerplanar embedding of $Q$ in which all nodes, and in particular $y$, lie on the
outer face. Draw the path $P$ in the outer face without crossings, and then connect both end-nodes of
$P$ to $y$ by drawing the two edges within the outer face so that they intersect the embedding of $Q$
only at $y$. This yields an outerplanar embedding of $S_{y,y'}$, as illustrated in~\cref{fig:outerplanar_obs}.

If $y \neq y'$ and $\{y,y'\}\notin E(Q)$, suppose for a contradiction that $S_{y,y'}$ is outerplanar.
Every node of $P$ has degree $2$ in $S_{y,y'}$. Smoothing a degree-$2$ node in an outerplanar graph preserves
outerplanarity: given an outerplanar embedding in which all nodes lie on the outer face, smoothing such a node yields
an embedding in which all remaining nodes still lie on the outer face. Hence we may smooth all nodes of $P$ in $S_{y, y'}$ while
preserving the outerplanarity of $S_{y, y'}$, thereby replacing the union of path $P$ and $\{\{x,y\},\{x',y'\}\}$ in $S_{y, y'}$ by the single edge $\{y,y'\}$. It follows that
$Q\cup\{\{y,y'\}\}$ has outerthickness $z$. Since $\{y,y'\}\notin E(Q)$, this adds an edge to $Q$, contradicting the
assumption that $Q$ is edge-maximal with outerthickness $z$. Therefore, $S_{y,y'}$ is not outerplanar.
\end{proof}

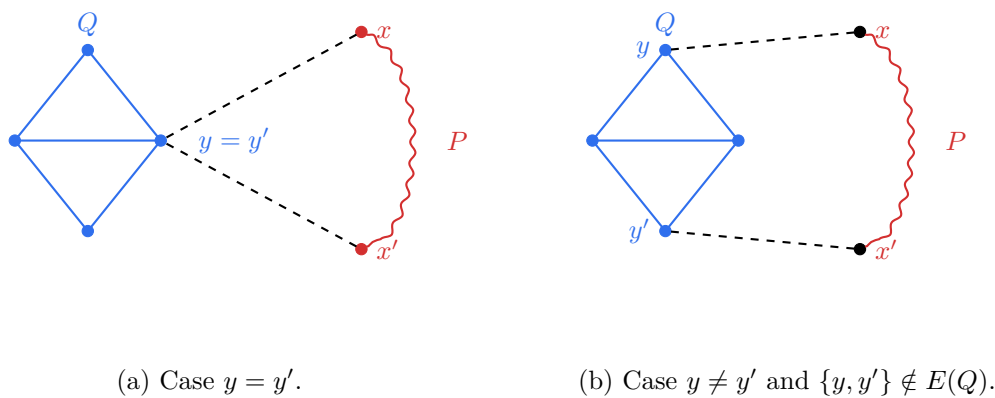
\begin{figure}[h]
    \centering
    \begin{tikzpicture}[
    vertex/.style={circle, draw, fill=black, inner sep=1.5pt, outer sep=0pt},
    edge/.style={thick, black},
    connedge/.style={thick, dashed, black},
    snakeedge/.style={thick, decorate, decoration={snake, amplitude=1pt, segment length=8pt}},
    qbox/.style={draw, thick, black, rounded corners=8pt, inner sep=12pt},
    label text/.style={font=\small},
    scale=0.8, transform shape=false
]

\begin{scope}[local bounding box=caseA]
    \coordinate (q_top_a) at (0, 1.5);
    \coordinate (q_left_a) at (-1.2, 0);
    \coordinate (q_right_a) at (1.2, 0);
    \coordinate (q_bottom_a) at (0, -1.5);

    \draw[edge, color=blue] (q_top_a) -- (q_left_a) -- (q_bottom_a) -- (q_right_a) -- (q_top_a);
    \draw[edge, color=blue] (q_left_a) -- (q_right_a);

    \node[vertex, color=blue] at (q_top_a) {};
    \node[vertex, color=blue] at (q_left_a) {};
    \node[vertex, color=blue, label={ [xshift=3mm]right:{\color{blue}$y=y'$}}] (Y_a) at (q_right_a) {};
    \node[vertex, color=blue] at (q_bottom_a) {};

    \node[anchor=south, inner sep=5pt, color=blue] at (q_top_a.south) {$Q$};

    \coordinate (p_x_a) at (4.5, 1.8);
    \coordinate (p_xp_a) at (4.5, -1.8);

    \node[vertex, color=red, label={right:\color{red}$x$}] (X_a) at (p_x_a) {};
    \node[vertex, color=red, label={right:\color{red}$x'$}] (XP_a) at (p_xp_a) {};
    
    \draw[snakeedge, color=red] (X_a) .. controls ++(1, -0.5) and ++(1, 0.5) .. (XP_a);
    
    \node[right=1cm, color=red] at ($(X_a)!0.5!(XP_a)$) {$P$};

    \draw[connedge] (X_a) -- (Y_a);
    \draw[connedge] (XP_a) -- (Y_a);

    \node[yshift=-2.8cm] at (2,-0.5) {(a) Case $y=y'$.};
\end{scope}

\begin{scope}[xshift=9.5cm, local bounding box=caseB]
    \coordinate (q_top_b) at (0, 1.5);
    \coordinate (q_left_b) at (-1.2, 0);
    \coordinate (q_right_b) at (1.2, 0);
    \coordinate (q_bottom_b) at (0, -1.5);

    \draw[edge,color=blue] (q_top_b) -- (q_left_b) -- (q_bottom_b) -- (q_right_b) -- (q_top_b);
    \draw[edge,color=blue] (q_left_b) -- (q_right_b);

    \node[vertex,color=blue, label={left:\color{blue}$y$}] (Y_b) at (q_top_b) {};
    \node[vertex,color=blue] at (q_left_b) {};
    \node[vertex,color=blue] at (q_right_b) {};
    \node[vertex,color=blue, label={left:\color{blue}$y'$}] (YP_b) at (q_bottom_b) {};

    \node[anchor=south, inner sep=5pt] at (q_top_b.south) {\color{blue}$Q$};

    \coordinate (p_x_b) at (3.2, 1.8);
    \coordinate (p_xp_b) at (3.2, -1.8);

    \node[vertex, label={right:\color{red}$x$}] (X_b) at (p_x_b) {};
    \node[vertex, label={right:\color{red}$x'$}] (XP_b) at (p_xp_b) {};
    
    \draw[snakeedge,color=red] (X_b) .. controls ++(1, -0.5) and ++(1, 0.5) .. (XP_b);
    
    \node[right=1cm,color=red] at ($(X_b)!0.5!(XP_b)$) {$P$};

    \draw[connedge] (X_b) -- (Y_b);
    \draw[connedge] (XP_b) -- (YP_b);

    \node[yshift=-2.8cm] at (2,-0.5) {(b) Case $y \ne y'$ and $\{y,y'\}\notin E(Q)$.};
\end{scope}

\end{tikzpicture}
    \caption{An illustration of $S_{y,y'}$ with $Q = K_4 -\{e\}$ and $z =1$.\label{fig:outerplanar_obs}}
\end{figure}

We are ready to prove the main lemma. 

\begin{lemma}\label{lem:kernel}
$G$ admits a 3-edge-coloring if and only if $G'$ has outerthickness at most $3$.
\end{lemma}

\begin{proof}
If $G$ admits a $3$-edge-coloring, then $E(G)$ can be partitioned into three matchings $M_1,M_2,$ and $M_3$. Fix $i\in[3]$ and $\ell\in\mathcal{L}$. Then the edges $e\coloneqq\{u,v\}\in M_i$ with $\varphi(e)=\ell$ give rise to a collection $T_{i, \ell}$ of node-disjoint triangles, each on nodes $\{u,v,w_\ell\}$, where the only node shared among these triangles is the common node $w_\ell$. Let $H$ have an edge partition into three outerplanar subgraphs $H_1,H_2, \mbox{ and }H_3$. By~\cref{obs:extrapath}(a), for each $i \in [3]$, the subgraph 
\[
H_i \cup \bigcup_{\ell \in \mathcal{L}} T_{i, \ell}
\]
is outerplanar. Thus, $G'$ has outerthickness at most $3$.

Otherwise, $G$ does not admit a $3$-edge-coloring. Suppose $E(G')$ can be partitioned into three outerplanar subgraphs $E_i$ for $i \in [3]$, the following claims hold.
\begin{claim}\label{clm:p3}
There exists some $i \in [3]$, some path of length $2$ in $G$ is entirely within an outerplanar subgraph $E_i$ for some $i\in[3]$.
\end{claim}
\begin{claimproof}
Now $E(G')$ is partitioned into three outerplanar subgraphs $E_i$ for $i \in [3]$. Restricting this partition to $E(G) \subseteq E(G')$, the induced parts $E(G)\cap E_i$
for $i\in[3]$ cannot all be matchings; otherwise, they would yield a $3$-edge-coloring of $G$.
Therefore, for some $i\in[3]$, the subgraph $E_i$ contains a path of length $2$ in $G$.
\end{claimproof}

\begin{claim}\label{clm:distinct}
For each node $v \in V(G)$, let $w_{v, j}$ for $j \in [3]$ be neighbor nodes of $v$ in $H$. Then, the edges $\{v, w_{v, j}\}$ for $j \in [3]$ have to be assigned to pairwise distinct outerplanar subgraphs $E_i$ for $i \in [3]$.
\end{claim}
\begin{claimproof}
The nodes $w_{v,j}$ for $j\in[3]$ are pairwise distinct, since the three edges of $G$ incident to $v$
receive distinct labels in $\mathcal{L}$ by the construction of $\varphi$. Recall that $w_1,w_2,\ldots,w_{\alpha}$ (and hence $w_{v,j}$
for $j\in[3]$) form an independent set.

Suppose that, for some $i\in[3]$, the subgraph $(V(G),E_i)$ contains at least two edges from
$\{\{v,w_{v,j}\}: j\in[3]\}$, say without loss of generality $\{v,w_{v,1}\}$ and $\{v,w_{v,2}\}$.
Let
\[
E' \coloneqq (E(H)\cap E_i)\cup\{\{v,w_{v,j}\}: j\in[2]\}.
\]
By \cref{obs:extrapath}(b), the graph $(V(H)\cup\{v\},E')$ is not outerplanar: apply the observation with
$(Q,E_1,P,y,y')_{\mathrm{\cref{obs:extrapath}}}=(H,E_i,(v),w_{v,1},w_{v,2})_{\mathrm{here}}$. On the other hand, $E'$ is a subgraph of the
outerplanar graph $E_i$, a contradiction.

Consequently, the three edges $\{v,w_{v,j}\}$ for $j \in [3]$ must be assigned to three distinct
outerplanar subgraphs $E_i$ for $i \in [3]$.
\end{claimproof}

\begin{claim}\label{clm:c3}
For each edge $e \coloneqq \{x, y\} \in E(G)$, the three edges $\{x, w_{\varphi(e)}\}$, $\{y, w_{\varphi(e)}\}$, $\{x, y\}$ have to be assigned to the same outerplanar subgraph $E_i$ for some $i \in [3]$.
\end{claim}
\begin{claimproof}
Since $E_1,E_2,E_3$ form a partition of $E(G)$, we may assume w.l.o.g.\ that $\{x,y\}\in E_1$.
Let
\[
W_x \coloneqq \{w_{x,1},\, w_{x,2},\, w_{\varphi(e)}\}
\qquad\text{and}\qquad
W_y \coloneqq \{w_{y,1},\, w_{y,2},\, w_{\varphi(e)}\}
\]
denote the sets of nodes in $H$ that are adjacent to $x$ and $y$ in $G'$, respectively.
By the construction of $\varphi$, each of $W_x$ and $W_y$ consists of three distinct nodes,
and $W_x\cap W_y=\{w_{\varphi(e)}\}$.
By \cref{clm:distinct}, there exist edges $\{x,w\}$ with $w\in W_x$ and $\{y,w'\}$ with $w'\in W_y$
that both lie in $E_1$.

If $w\neq w'$, then by \cref{obs:extrapath}(b) the subgraph induced by the union of the path $P \coloneqq (w, x, y, w')$  and $E_1 \cap E(H)$ is not outerplanar: apply the observation with
$(Q,E_1,P,y,y')_{\mathrm{\cref{obs:extrapath}}}=(H,E_1,(x, y),w,w')_{\mathrm{here}}$.

Therefore $w=w'=w_{\varphi(e)}$, or equivalently, the three edges
$\{x,w_{\varphi(e)}\}$, $\{y,w_{\varphi(e)}\}$, and $\{x,y\}$ are all assigned to $E_1$.
\end{claimproof}

By~\cref{clm:p3}, we assume w.l.o.g. that $E_1$ contains the path of length $2$, denoted by $(a, b, c)$ with $a, b, c \in V(G)$ and $e_1 \coloneqq \{a, b\}, e_2 \coloneqq \{b, c\} \in E(G)$. By~\cref{clm:c3}, $E_1$ contains the four edges $\{a, w_{\varphi(e_1)}\}$, $\{b, w_{\varphi(e_1)}\}$, $\{b, w_{\varphi(e_2)}\}$, $\{c, w_{\varphi(e_2)}\}$ as well. However, by \cref{clm:distinct}, $\{b, w_{\varphi(e_1)}\}$ and $\{b, w_{\varphi(e_2)}\}$ have to be assigned to different $E_i$'s, a contradiction.

As a result, $E(G')$ cannot be partitioned into three outerplanar subgraphs $E_i$ for $i \in [3]$ or, equivalently, $G'$ has outerthickness greater than $3$. This completes the proof.
\end{proof}

\subsection{Proofs of Deferred Claims}


\begin{proposition}\label{pro:varphi}
Fix an integer $k \ge 3$. Let $G$ be a $k$-regular graph. There exists a labeling function $\varphi: E(G) \to \mathcal{L}$ with $|\mathcal{L}| \le 2k(k-1)+1$, computable in polynomial time, such that the edges on any path in $G$ of length at most $3$ receive pairwise distinct labels.
\end{proposition}
\begin{proof}
In a $k$-regular graph $G$, fix an edge $e=\{u,v\}$. Any other edge $f\neq e$ that can appear together with $e$ on a path of length at most $3$ must either be incident to $u$ or $v$ (at most $2(k-1)$ choices), or be incident to one of the four neighbors of $u$ and $v$ (at most $2(k-1)^2$ additional choices). Hence, for each edge $e$, the number of edges $f\neq e$ that can be contained with $e$ in some path of length at most $3$ is at most $2k(k-1)$. Therefore, we can greedily assign to each edge $e$ a label from $\{1,2,\ldots,2k(k-1)+1\}$ so that no such edge $f$ receives the same label as $e$; since at most $2k(k-1)$ labels are forbidden at each step, at least one label remains available. The running time of the above greedy labeling is linear in the input size.
\end{proof}

\begin{proposition}\label{pro:H}
    For each integer $|\mathcal{L}| \ge 1$, there exists a constant $C$ such that every $C$-node graph $H$ that is edge-maximal with outerthickness $3$ contains an independent set of size at least $\alpha \ge |\mathcal{L}|$. Moreover, such a graph $H$ and an independent set of size $\alpha$ in $H$ can be computed in constant time.
\end{proposition}
\begin{proof}
Fix $|\mathcal{L}| \ge 1$ and choose a constant $C$ such that
\begin{equation}\label{eqn:complete}
  \binom{C}{2} > 3(2C-3).
\end{equation}
Every $C$-node outerplanar graph has at most $2C-3$ edges, and hence any graph of outerthickness at most $3$ has at most $3(2C-3)$ edges. Choosing $C$ so that $\binom{C}{2} > 3(2C-3)$ implies that $K_C$ has outerthickness greater than $3$. Let $H$ be a subgraph of $K_C$ with the maximum number of edges among all subgraphs of $K_C$ whose outerthickness is at most $3$; such an $H$ exists since the empty graph has outerthickness $0$. By maximality, adding any edge from $E(K_C)\setminus E(H)$ to $H$ yields a graph of outerthickness greater than $3$. Moreover, adding a single edge can increase outerthickness by at most $1$. Therefore, $H$ is edge-maximal with outerthickness (exactly) $3$.

By the Caro--Wei bound~\cite{Caro79,Wei81}, every $C$-node $m$-edge graph has an independent set of size at least
$\left\lceil \frac{C^2}{C+2m}\right\rceil$.
Since $H$ has $m \le 3(2C-3)$ edges for our choice of $C$, we obtain that $H$ has an independent set of size at least
\begin{equation}\label{eqn:cw}
  \left\lceil \frac{C^2}{C+6(2C-3)}\right\rceil > |\mathcal{L}| \mbox{ for some sufficiently large } C.
\end{equation}
If we set $C > 13 |\mathcal{L}|$, then \cref{eqn:complete,eqn:cw} both hold. 

Finally, since $C$ is a constant, we can compute such an $H$ and an independent set of size at least $|\mathcal{L}|$ in constant time by brute force, using an outerplanarity testing algorithm~\cite{Mitchell79}.
\end{proof}


\section{NP-hardness of $\mathcal{F}$-Thickness}
\label{sec:main}

We generalize the NP-hardness proof for outerthickness in \cref{sec:kernel} to every proper graph class $\mathcal{F}$ whose membership is decidable and that satisfies all of the following conditions:
\begin{enumerate}[(a)]
    \item $\mathcal{F}$ is closed under topological minors,
    \item $\mathcal{F}$ is closed under $1$-sums, and
    \item $\mathcal{F}$ contains a  cycle of length $3$.  
\end{enumerate}

By Condition~(a), $\mathcal{F}$ is monotone. Thus, any covering of the edge set of an input graph for the problem $P_{\mathcal{F}}$ can be converted into an edge partition by deleting overlaps; hence, in this setting, covering and partitioning are equivalent. \textbf{Thus, in this section, we treat $P_{\mathcal{F}}$ as a partition problem:} Fix a proper graph class $\mathcal{F}$ whose membership can be decided by algorithm $\mathcal{M}_{\mathcal{F}}$.
We define the \emph{$\mathcal{F}$-thickness} of an undirected simple graph $G$, denoted by $\theta_{\mathcal{F}}(G)$, to be the minimum integer $k$ such that $E(G)$ admits a partition into at most $k$ subsets $E_1, E_2, \dots, E_k$ with $(V(G), E_i) \in \mathcal{F}$ for all $i \in [k]$. Thus, the problem $P_{\mathcal{F}}$ asks, given an undirected simple graph $G$ and an integer $k$, whether $\theta_{\mathcal{F}}(G) \le k$.

Our NP-hardness reduction proceeds as follows. Fix an integer $k \ge 3$. Given an input instance $G$ of \textsc{Edge-Coloring} on $k$-regular graphs, defined in~\cref{thm:edgecolor}. We compute the following label function $\mathcal{L}$ and an auxiliary graph $H$ in polynomial time. Then, we obtain a supergraph $G'$ from joining $G$ and $H$ using label function $\mathcal{L}$ such that $G$ admits a $k$-edge-coloring if and only if $G'$ has $\theta_\mathcal{F}(G') \le k$ (\cref{lem:general}). This proves the NP-hardness claimed in \cref{thm:main} for every fixed integer $k \ge 3$.

\begin{itemize}
    \item Let $\varphi: E(G)\to \mathcal{L}$ be a labeling function, where $\mathcal{L} = \{1, 2, \ldots, |\mathcal{L}|\}$ is a constant-size (not necessarily minimum-size) label set, such that for every path in $G$ of length at most $3$, the edges on the path receive pairwise distinct labels under $\varphi$. The existence of such a function $\varphi$, together with a construction running in time polynomial in the input size of $G$, is guaranteed by \cref{pro:varphi}. The resulting label set $\mathcal{L}$ satisfies $|\mathcal{L}| = 2k(k-1)+1 = O(1)$, as desired.

\medskip

    \item We pick a sufficiently large constant $C$ such that every $C$-node graph with $\mathcal{F}$-thickness $k$ contains an (not necessarily maximum) independent set of $\alpha \ge |\mathcal{L}|$ nodes. Then we construct a $C$-node graph $H$ that is edge-maximal with $\mathcal{F}$-thickness $k$. Given the promise of $C$, let $w_1, w_2, \ldots, w_\alpha$ be the $\alpha$ nodes in an independent set in $H$. The existence of $H$, as well as a construction running in time polynomial in the input size of $G$, is shown in \cref{lem:sparse}. 
\end{itemize}
    
    \begin{lemma}\label{lem:sparse}
    For each integer $|\mathcal{L}| \ge 1$, there exists a constant $C$ such that every $C$-node graph $H$ that is edge-maximal with $\mathcal{F}$-thickness $k$ contains an independent set of at least $\alpha \ge |\mathcal{L}|$ nodes. Such an $H$ and an independent set in $H$ of $\alpha$ nodes can be computed in constant time.
    \end{lemma}
    \begin{proof}
    Since $\mathcal{F}$ is proper and closed under topological minors, there exists an integer $t$ such that $K_t \notin \mathcal{F}$. By a theorem of Bollob\'as and Thomason~\cite{BollobasThomason1998}, there is a constant $D_t$ such that every graph $F$ with more than $D_t|V(F)|$ edges contains a subdivision of $K_t$, and hence does not belong to $\mathcal{F}$. Therefore, every $F \in \mathcal{F}$ satisfies $|E(F)| \le D_t|V(F)|$.

    Fix $|\mathcal{L}| \ge 1$ and choose a constant $C$ such that
    \begin{equation}\label{eqn:kt}
        \binom{C}{2} > k(D_t C).
    \end{equation}

\begin{claim}\label{clm:plusone}
For every $F \in \mathcal{F}$ and $e \notin E(F)$, $\theta_{\mathcal{F}}(F \cup \{e\}) \le \theta_{\mathcal{F}}(F)+1$.
\end{claim}
\begin{claimproof}
Since $C_3 \in \mathcal{F}$ (Condition~(c)) and $\mathcal{F}$ is closed under topological minors (Condition~(a)), we have $P_2 \in \mathcal{F}$. Thus, we have $\theta_{\mathcal{F}}(F \cup \{e\}) \le \theta_{\mathcal{F}}(F) + \theta_{\mathcal{F}}(\{e\})
\le \theta_{\mathcal{F}}(F) + 1$.
\end{claimproof}

    Since $\theta_{\mathcal{F}}(\emptyset)=0$, the family
    \[
    \mathcal{S}\coloneqq \{S \subseteq K_C : \theta_{\mathcal{F}}(S)\le k\}
    \]
    is nonempty. Let $S^*$ be a member of $\mathcal{S}$ with maximum number of edges.
    Thus $S^*$ is edge-maximal with $\mathcal{F}$-thickness at most $k$. To see why, for every edge
    $e\in E(K_C)\setminus E(S^*)$ we must have $\theta_{\mathcal{F}}(S^* \cup \{e\})>k$,
   otherwise $S^*\cup \{e\}\in\mathcal{S}$ contradicts the choice of $S^*$.  Moreover, by~\cref{clm:plusone}, we have $\theta_{\mathcal{F}}(S^*) = k$.

By the Caro--Wei bound~\cite{Caro79,Wei81}, every $C$-node $m$-edge graph has an independent set of size at least
$\left\lceil \frac{C^2}{C+2m}\right\rceil$.
Since $H$ has $m \le k(D_t C)$ edges for our choice of $C$, we obtain that $H$ has an independent set of size at least
\begin{equation}\label{eqn:cw2}
  \left\lceil \frac{C^2}{C+2k(D_t C))}\right\rceil > |\mathcal{L}| \mbox{ for some sufficiently large } C.
\end{equation}
If we set $C > (1+2kD_t) |\mathcal{L}|$, then \cref{eqn:kt,eqn:cw2} both hold.

Finally, since $C$ is a constant, we can compute such an $H$ and an independent set of size at least $|\mathcal{L}|$ in constant time by brute force, using the algorithm $\mathcal{M}_{\mathcal{F}}$ to decide the membership in $\mathcal{F}$. 
\end{proof}

\begin{observation}\label{obs:extrapath-general}
Fix an integer $z \ge 1$. Let $Q$ be a graph that is edge-maximal with $\mathcal{F}$-thickness $z$, and let $E_1, E_2, \ldots, E_z$ be an edge-partition of $E(Q)$ into $z$  subgraphs in $\mathcal{F}$.
Let $P$ be a path with endpoints $x$ and $x'$
such that $V(P)\cap V(Q)=\varnothing$. For any nodes $y, y'\in V(Q)$, define
\[
S_{y, y'} \coloneqq E_1 \cup P \cup \bigl\{\{x,y\},\{x',y'\}\bigr\}.
\]

If $y\ne y'$ and $\{y,y'\}\notin E(Q)$, then $S_{y,y'}$ does not belong to $\mathcal{F}$.
\end{observation}
\begin{proof}
If $y \neq y'$ and $\{y,y'\}\notin E(Q)$, suppose for contradiction that $S_{y,y'}$ belongs to $\mathcal{F}$.
Every node of $P$ has degree $2$ in $S_{y,y'}$. Since $\mathcal{F}$ is closed under topological minors (Condition (a)), smoothing a degree-$2$ node preserves the membership in $\mathcal{F}$.
Hence we may smooth all nodes of $P$ in $S_{y, y'}$ while
preserving that $S_{y, y'} \in \mathcal{F}$, thereby replacing 
the union of path $P$ and $\{\{x,y\},\{x',y'\}\}$ in $S_{y, y'}$
by the single edge $\{y,y'\}$. It follows that
$Q\cup\{\{y,y'\}\}$ has $\mathcal{F}$-thickness $z$. Since $\{y,y'\}\notin E(Q)$, this adds an edge to $Q$, contradicting the
assumption that $Q$ is edge-maximal with $\mathcal{F}$-thickness $z$. Therefore, $S_{y,y'}$ does not belong to $\mathcal{F}$.
\end{proof}

We initialize $G'$ by taking copies of $G$ and $H$ with disjoint node sets, and then identifying $G'$ with their union. Then, for each edge $e\coloneqq\{u,v\}\in E(G)$, we add the edges $\{u,w_{\varphi(e)}\}$ and $\{v,w_{\varphi(e)}\}$ to $G'$. Note that $G'$ remains a simple graph: since any two edges incident to the same node in $G$ receive distinct labels under $\varphi$, no parallel edges are created.

We are ready to prove the key lemma. 

\begin{lemma}\label{lem:general}
$G$ admits a $k$-edge-coloring if and only if $G'$ has $\theta_{\mathcal{F}}(G') \le k$.
\end{lemma}
\begin{proof}
If $G$ admits a $k$-edge-coloring, then $E(G)$ can be partitioned into $k$ matchings $M_i$ for $i \in [k]$. Fix $i^*\in[k]$ and $\ell^* \in\mathcal{L}$. The edges $e\coloneqq\{u,v\}\in M_{i^*}$ with $\varphi(e)=\ell^*$ give rise to a collection $T_{i^*, \ell^*}$ of node-disjoint triangles, each on nodes $\{u,v,w_{\ell^*}\}$, where the only node shared among these triangles is the common node $w_{\ell^*}$. Since $H$ is edge-maximal with $\mathcal{F}$-thickenss $k$, $H$ has an edge partition into $k$ subgraphs in $\mathcal{F}$, denoted by $H_i$ for $i \in [k]$. Since $\mathcal{F}$ is closed under $1$-sums (Condition (b)) and $\mathcal{F}$ contains $C_3$ (Condition (c)), for each $i \in [k]$, the subgraph 
\[
H_i \cup \bigcup_{\ell \in \mathcal{L}} T_{i, \ell}
\]
is contained in $\mathcal{F}$. To see why, for every $\ell\in\mathcal{L}$, we attach each triangle in $T_{i,\ell}$ to $H_i$ by
taking a $1$-sum that identifies the node $w_\ell$. Repeating this $1$-sum operation
for all triangles in $T_{i,\ell}$, and then for all $T_{i, \ell}$ with $\ell\in\mathcal{L}$, yields exactly
$H_i \cup \bigcup_{\ell \in \mathcal{L}} T_{i,\ell}$, which therefore lies in $\mathcal{F}$. This argument holds for every $i \in [k]$, so $G'$ has $\mathcal{F}$-thickness at most $k$.

\medskip

Otherwise, $G$ does not admit a $k$-edge-coloring. Suppose $E(G')$ can be partitioned into $k$ subgraphs in $\mathcal{F}$, denoted by $E_i$ for $i \in [k]$, the following claims hold.
\begin{claim}\label{clm:p3-general}
There exists some $i \in [k]$, some path of length $2$ in $G$ is entirely within a subgraph $E_i$ (in $\mathcal{F}$) for some $i\in[k]$.
\end{claim}
\begin{claimproof}
Now $E(G')$ is partitioned into $k$ subgraphs $E_i$ in $\mathcal{F}$ for $i \in [k]$. 
Restricting this partition to $E(G) \subseteq E(G')$, the induced parts $E(G)\cap E_i$
for $i\in[k]$ cannot all be matchings; otherwise, they would yield a $k$-edge-coloring of $G$.
Therefore, for some $i\in[k]$, the subgraph $E_i$ contains a path of length $2$ in $G$.
\end{claimproof}

\begin{claim}\label{clm:distinct-general}
For each node $v \in V(G)$, let $w_{v, j}$ for $j \in [k]$ be neighbor nodes of $v$ in $H$. Then, the edges $\{v, w_{v, j}\}$ for $j \in [k]$ have to be assigned to pairwise distinct subgraphs $E_i$ (in $\mathcal{F}$) for $i \in [k]$.
\end{claim}
\begin{claimproof}
The nodes $w_{v,j}$ for $j\in[k]$ are pairwise distinct, since the $k$ edges in $G$ incident to $v$
receive distinct labels in $\mathcal{L}$ by the construction of $\varphi$. Recall that $w_1,w_2,\ldots,w_{\alpha}$ (and hence $w_{v,j}$
for $j\in[k]$) form an independent set.

Suppose that, for some $i \in [k]$, the subgraph $E_i$ contains at least two edges from
$\{\{v,w_{v,j}\} : j \in [k]\}$, say without loss of generality $\{v,w_{v,1}\}$ and $\{v,w_{v,2}\}$.
Let
\[
E' \coloneqq (E(H)\cap E_i)\cup\{\{v,w_{v,j}\}: j\in[2]\}.
\]

By \cref{obs:extrapath-general}, the graph $(V(H)\cup\{v\},E')$ does not belong to $\mathcal{F}$: indeed, apply the
observation with
\[
(Q,E_1,P,y,y')_{\mathrm{\cref{obs:extrapath-general}}}=(H,E_i,(v),w_{v,1},w_{v,2})_{\mathrm{here}}.
\]
On the other hand, $E'$ is a subgraph of $E_i$, which lies in $\mathcal{F}$. This contradicts
Condition~(a), since $\mathcal{F}$ is closed under topological minors (and hence under taking subgraphs).

Consequently, the $k$ edges $\{v,w_{v,j}\}$ for $j \in [k]$ must be assigned to $k$ distinct subgraphs $E_i$ for $i \in [k]$.
\end{claimproof}

\begin{claim}\label{clm:c3-general}
For each edge $e \coloneqq \{x, y\} \in E(G)$, the three edges $\{x, w_{\varphi(e)}\}$, $\{y, w_{\varphi(e)}\}$, $\{x, y\}$ have to be assigned to the same subgraph $E_i$ (in $\mathcal{F}$) for some $i \in [k]$.
\end{claim}
\begin{claimproof}
Since $E_i$ for $i \in [k]$ form a partition of $E(G)$, we may assume w.l.o.g.\ that $\{x,y\}\in E_1$.
Let
\[
W_x \coloneqq \{w_{x,1},\ldots , w_{x,k-1},\, w_{\varphi(e)}\}
\text{ and }
W_y \coloneqq \{w_{y,1},\ldots, w_{y,k-1},\, w_{\varphi(e)}\}
\]
denote the sets of nodes in $H$ that are adjacent to $x$ and $y$ in $G'$, respectively.
By the construction of $\varphi$, each of $W_x$ and $W_y$ consists of $k$ distinct nodes,
and $W_x\cap W_y=\{w_{\varphi(e)}\}$.
By \cref{clm:distinct-general}, there exist edges $\{x,w\}$ with $w\in W_x$ and $\{y,w'\}$ with $w'\in W_y$
that both lie in $E_1$.

If $w\neq w'$, then by \cref{obs:extrapath-general} the subgraph induced by the union of the path $P \coloneqq (w, x, y, w')$  and $E_1 \cap E(H)$ does not belong to $\mathcal{F}$: apply the observation with
\[
(Q,E_1,P,y,y')_{\mathrm{\cref{obs:extrapath-general}}}=(H,E_1,(x, y),w,w')_{\mathrm{here}}.
\]

Therefore $w=w'=w_{\varphi(e)}$, or equivalently, the three edges
$\{x,w_{\varphi(e)}\}$, $\{y,w_{\varphi(e)}\}$, and $\{x,y\}$ are all assigned to $E_1$.
\end{claimproof}

By~\cref{clm:p3-general}, we assume w.l.o.g. that $E_1$ contains the path of length $2$, denoted by $(a, b, c)$ with $a, b, c \in V(G)$ and $e_1 \coloneqq \{a, b\}, e_2 \coloneqq \{b, c\} \in E(G)$. By~\cref{clm:c3-general}, $E_1$ contains the four edges $\{a, w_{\varphi(e_1)}\}$, $\{b, w_{\varphi(e_1)}\}$, $\{b, w_{\varphi(e_2)}\}$, $\{c, w_{\varphi(e_2)}\}$ as well. However, by \cref{clm:distinct-general}, $\{b, w_{\varphi(e_1)}\}$ and $\{b, w_{\varphi(e_2)}\}$ have to be assigned to different $E_i$'s, a contradiction.

As a result, $E(G')$ cannot be partitioned into $k$ subgraphs $E_i$ in $\mathcal{F}$ for $i \in [k]$ or, equivalently, $G'$ has $\mathcal{F}$-thickness greater than $k$. This completes the proof.
\end{proof}

As a result, \cref{thm:main} follows.


\section{Concluding Remarks}\label{sec:con}

Finally, we remark on the case $k=2$. It is known that every planar graph has outerthickness at most $2$ due to Gon{\c{c}}alves (2005)~\cite{Goncalves05},
and hence \textsc{Outerthickness}$(G,2)$ for planar graphs contains only Yes-instances. For general graphs, however, the complexity of
\textsc{Outerthickness}$(G,2)$ remains, to the best of our knowledge, an intriguing open problem. More broadly, in the
context of our hardness framework in \cref{thm:main}, we are not aware of any proper graph class $\mathcal{F}$ satisfying
Conditions~(a)--(c) for which the corresponding problem $P_{\mathcal{F}}$ with $k=2$ admits a polynomial-time algorithm.

It is also worth noting that Conditions~(a)--(c) in \cref{thm:main} are necessary, in the sense that relaxing any of them allows for graph classes $\mathcal{F}$ where $P_{\mathcal{F}}$ is solvable in polynomial time due to Edmonds \cite{Edmonds65}. For instance, if $\mathcal{F}$ is the class of forests, Condition~(c) fails, and the problem is in P. Similarly, if $\mathcal{F}$ is the class of pseudoforests, Condition~(b) fails, and $P_{\mathcal{F}}$ also is polynomial-time solvable. Finally, if $\mathcal{F}$ is the class of graphs of bounded arboricity, Condition~(a) is not satisfied\footnote{For example, the 1-subdivision of $K_5$ has arboricity at most 2, while its topological minor $K_5$ has arboricity 3.}, and the corresponding partitioning problem is also in P.
\bibliography{ref}

\end{document}